\documentclass[a4paper,11pt]{amsart}
\usepackage[utf8]{inputenc}
\usepackage{amsaddr}
\usepackage{calrsfs}
\usepackage{graphicx,color}
\usepackage[english]{babel}
\usepackage{graphicx}
\usepackage{amsmath}
\usepackage{epstopdf}
\usepackage{float}
\usepackage{color}
\usepackage{cite}

\usepackage{amssymb}
\usepackage[final]{pdfpages}
\usepackage{amsmath}
\usepackage{amstext}
\usepackage{amssymb}
\usepackage{graphicx}
\usepackage[utf8]{inputenc}

\numberwithin{equation}{section}

\theoremstyle{definition}

\theoremstyle{plain}
\newtheorem{thm}{Theorem}[section]

\newtheorem{cor}{Corollary}[section]

\theoremstyle{definition}

\newcommand{\E}{\mathbb{E}}

\begin{document}
\title[Fractional short rate ]
{Equity warrant pricing under subdiffusive fractional Brownian motion of the short rate}

\date{\today}

\author[Shokrollahi]{Foad Shokrollahi}
\address{Department of Mathematics and Statistics, University of Vaasa, P.O. Box 700, FIN-65101 Vaasa, FINLAND}
\email{foad.shokrollahi@uva.fi}

\author[Magdziarz]{Marcin Magdziarz}
\address{Hugo Steinhaus Center, Faculty of Pure and Applied Mathematics, Wrocław University of Science and Technology, Wyspiańskiego 27, Wrocław, 50-370, Poland}
\email{marcin.magdziarz@pwr.edu.pl}

\begin{abstract}

In this paper we propose an extension of the Merton model. We apply the subdiffusive mechanism to analyze equity warrant in a fractional Brownian motion  environment, when the short rate follows the subdiffusive fractional Black-Scholes model. We obtain the pricing formula for zero-coupon bond in the introduced model and derive the partial differential equation with appropriate boundary conditions for the valuation of equity warrant. Finally, the pricing formula for equity warrant is provided under subdiffusive fractional Black-Scholes model of the short rate.
\end{abstract}


\maketitle
\section{Introduction}\label{sec:0}

Analysis of financial data displays that various processes in finance show certain periods in which they are constant \cite{janczura2009subdynamics}. Analogous property is observed in physical system with subduffusion. The constant periods of financial processes correspond to the trapping event in which the subdiffusive particle is motionless \cite{metzler2000random,metzler1999anomalous,eliazar2004spatial,dybiec2010subordinated}. The mathematical interpretation of subdiffusion is in terms of Fractional Fokker Planck equation $(FFPE)$. This equation was introduced from the continuous time random walk $(CTRW)$ strategy with heavy-tailed waiting times \cite{metzler1999anomalous,metzler2000random,sokolov2005diffusion}, later used as a substantial tool to evaluate complex systems with slow dynamics. In this paper, we use the fractional Black-Scholes $(FBS)$ model and the subdiffusive mechanism to better describe the dynamics observed in financial markets. We use similar strategy as in \cite{magdziarz2009black, sokolov2002solutions}, where the objective time $t$ was replaced by the inverse $\alpha$-stable subordinator $T_{\alpha}(t)$ in the $FBS$ model. $T_{\alpha}(t)$ corresponds to the fat-tailed waiting times in the underlying $CTRW$ and adds the constant periods to the dynamics of financial assets.  Then, the dynamics of asset price $V(t)$ is given by the following subdiffusive $FBS$ equation

\begin{eqnarray}
dV(T_{\alpha}(t))&=&\mu_v V(T_{\alpha}(t))d(T_{\alpha}(t))+\sigma_v V(T_{\alpha}(t))dB_1^{H}(T_{\alpha}(t)),
\label{eq:1}
\end{eqnarray}
where $\mu_v, \sigma_v$ are constant, $B_1^{H}$ is the fractional Brownian motion $(FBM)$ \cite{biagini2008stochastic}  with Hurst parameter $H \in[\frac{1}{2}, 1)$. $T_{\alpha}(t)$ is the inverse $\alpha$-stable subordinator with $\alpha\in (0,1)$ defined as
\begin{eqnarray}
T_{\alpha}(t)=\inf\{\tau>0: U_{\alpha}(\tau)>t\},
\label{eq:2}
\end{eqnarray}
Here $\{U_{\alpha}(t)\}_{t\geq0}$ is a $\alpha$-stable L\'evy process with nonnegative increments and Laplace transform: $E\left(e^{-uU_{\alpha}(t)}\right)=e^{-t u^{\alpha}}$ \cite{janicki1993simulation,gu2012time,wang2012continuous,hahn2011fokker}. When $\alpha\uparrow 1$, $T_{\alpha}(t)$ degenerates to $t$.

\begin{figure}[H]
  \centering
          \includegraphics[width=1\textwidth]{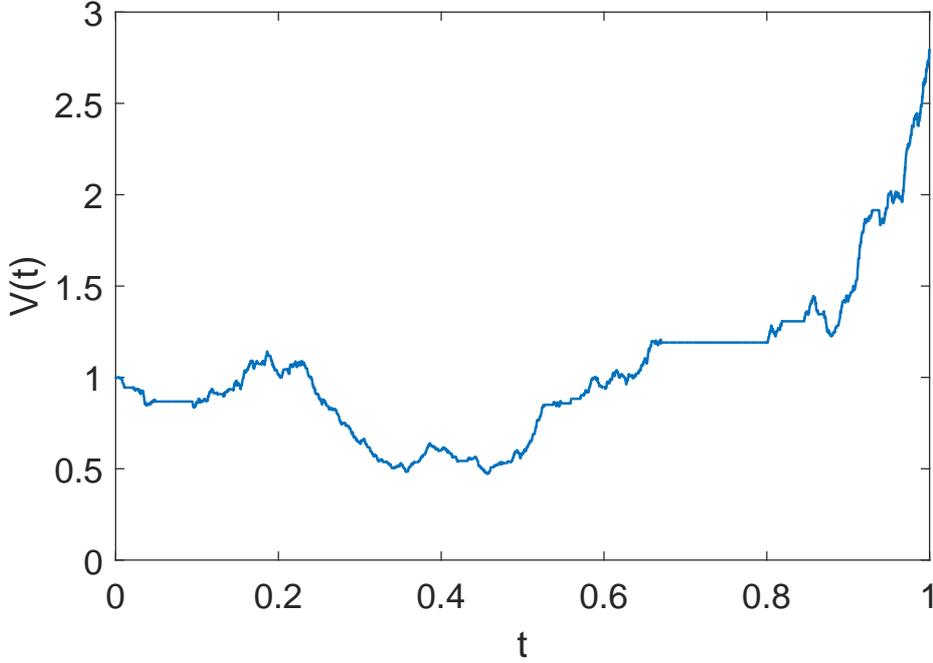}
  \caption{Typical trajectory of the asset price from formula (\ref{eq:1}). The parameters are: $\mu_v=\sigma_v=V(0)=1, H=0.7,\alpha = 0.9$.}
\label{fig1}
\end{figure}
However, all the above mentioned papers assume that the short rate is constant during the life of an option. But in reality the short rate evolves randomly over time. Hence, in order to take into account the stochastic short rate, we assume in this paper that the short rate follows the subdiffusive equation:
\begin{eqnarray}
dr(T_{\alpha}(t))&=&\mu_rd(T_{\alpha}(t))+\sigma_rdB_2^{H}(T_{\alpha}(t)).
\label{eq:3}
\end{eqnarray}
Here $\mu_r, \sigma_r$ are some constants, $B_2^{H}$ is a $FBM$ with Hurst parameter $H \in[\frac{1}{2}, 1)$ and $T_{\alpha}(t)$ is assumed to be independent of $B_2^{H}$. Moreover, $B_2^H$ and $B_1^1$  are two dependent $FBMs$ with correlation coefficient $\rho$. Additionally $T_{\alpha}(t)$ is assumed to be independent of $B_1^{H}$ and $B_2^{H}$.

\begin{figure}[H]
  \centering
          \includegraphics[width=1\textwidth]{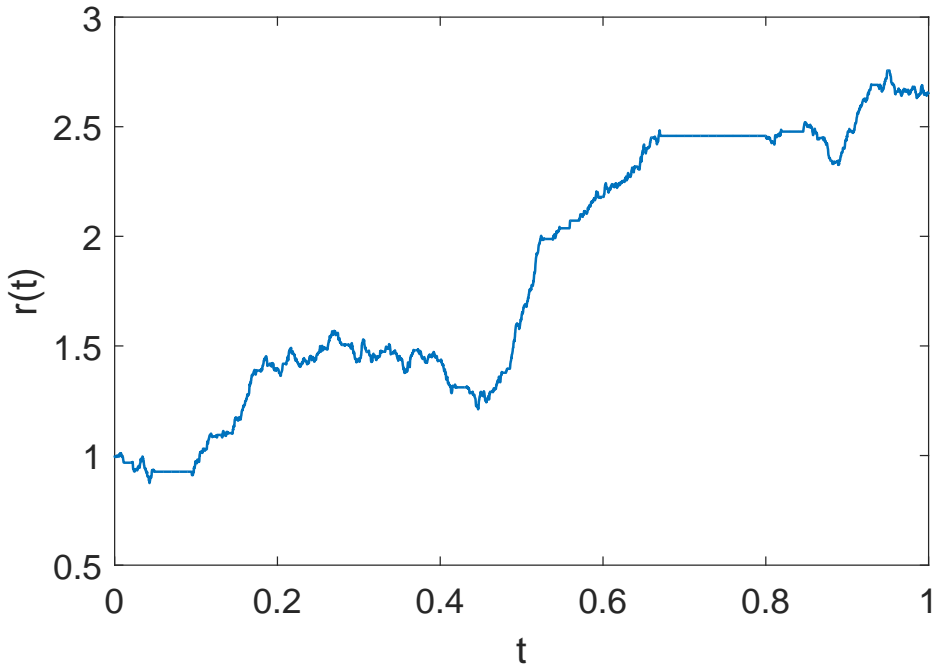}
  \caption{Typical trajectory of the short rate from formula (\ref{eq:3}) corresponding to the trajectory from Fig. \ref{fig1}. The parameters are: $\mu_r=\sigma_r=r(0)=1, H=0.7, \alpha = 0.9, \rho=0.5$.}
\label{fig2}
\end{figure}

The ongoing financial empirical evidences show that the mostly used Gaussian or, more general, Markov models may not be sufficient to capture the market structure for short interest rates \cite{vasicek1977equilibrium,cox2005theory,siu2010bond,hull1990pricing,elliott2002interest,chacko2002pricing}. One reason for this may be the fact that short rates, which are driven by macroeconomic variables, like domestic gross products, supply and demand rates or volatilities exhibit long range dependence, cannot be captured by Markov models \cite{ding1993long,tabak2005long,huang1995fractal}.
Motivated by the fact that the bond market and interest rates often, exhibit long-range dependence, in this paper, we incorporate the long memory nature of the short rate in our valuation model and derive explicit formulas for equity warrants when the short rate follows the subdiffusive fractional Brownian process.

Since the subdiffusive fractional Brownian process is a well-developed model, it is an efficient tool to capture the behavior of interest rates \cite{shokrollahi2016pricing,shokrollahi2018evaluation,shokrollahi2018subdiffusive,shokrollahi2020valuation} .
In what follows, we state some basic assumptions that will be used in this paper.
Given a risk neutral probability measure $(\Omega, \mathcal{F}, \mathbb{Q})$ on which the fractional Brownian motion is defined, we present some “ideal conditions” in the market for the firm value and for the equity warrant:
\begin{enumerate}
\item [(i)] There are no transaction costs or taxes and all securities are perfectly divisible;

\item[(ii)] Dividends are not paid during the lifetime of the outstanding warrants, and the sequential exercise of the warrants is not optimal for warrant holders;

\item[(iii)] The warrant-issuing firm is an equity firm with no outstanding debt;

\item[(iv)] There are no signaling effects associated with issuing warrants. The market neither reacts positively nor negatively to the information that the firm will issue (or has issued, respectively) warrants;

\item[(v)] Suppose that the firm has $N$ shares of common stocks and $M$ shares of equity warrants outstanding. Each warrant entitles the owner to receive $k$ shares of stocks at time $T$ upon payment of $X$.
\end{enumerate}

Assumptions (i)–(ii) are the standard assumptions in the Black–Scholes environment. Assumption (iii) implies that stocks and equity warrants are the only sources of financing that are issued by the firm. Assumption (iv) is a necessary condition for our analysis. This assumption implies that both the stock price and the firm’s value remain unaffected by warrants issuance. Actually, this condition could be achieved, see \cite{hanke2002consistent}. Assumption (v) sets some notations for the pricing model. Throughout this research, we assume that $\alpha\in(\frac{1}{2}, 1)$ and $\alpha+\alpha H>1$ (see, (\cite{gu2012time} and \cite{shokrollahi2018subdiffusive}). 

 The rest of this work is: In Section \ref{sec:1}, the pricing formula for a risk free zero-coupon bond is derived. In Section \ref{sec:2}, the corresponding $FBS$ is obtained using delta hedging strategy for equity warrant. Further, in Section \ref{sec:2} the value of equity warrant is presented when the dynamics of asset price follows a subdiffusive fractional Browian motion.

\section{Pricing model for a zero-coupon bond}\label{sec:1}
In this section, we assume that the short rate $r(t)$ satisfy Equation (\ref{eq:3}). Then, we obtain the pricing formula for zero-coupon bond $P(r, t, T)$. Here, $P(r, T, T)=1$, that is, the zero coupon bond will pay $1$ dollar at expiry date $T$.
\begin{thm}
The price of a zero-coupon bond with maturity $t\in [0, T]$ in the fractional Black-Scholes model is given by
\begin{eqnarray}
P(r, t, T)=e^{-rf_2(\tau)+f_1(\tau)},
\label{eq:18}
\end{eqnarray}
where $\tau=T-t$ and
\begin{eqnarray}
f_1(\tau)&=&\frac{H\sigma_r^2}{(\Gamma(\alpha))^{2H}}\int_0^{\tau}(T-v)^{2\alpha H-1}v^2dv\nonumber\\
&&-\frac{2H\mu_r}{(\Gamma(\alpha))^{2H}}\int_0^{\tau}(T-v)^{\alpha-1}vdv,\label{eq:16-1}\\
f_2(\tau)&=&\tau.
\label{eq:17}
\end{eqnarray}

\end{thm}

\begin{proof}
By the Taylor series expansion, we can get

\begin{eqnarray}
P(r+\Delta r, t+\Delta t)&=& P(r,t,T)+\frac{\partial P}{\partial r}\Delta r+\frac{\partial P}{\partial t}\Delta t\nonumber\\
&&+\frac{1}{2}\frac{\partial^2 P}{\partial r^2}(\Delta r)^2++\frac{1}{2}\frac{\partial^2 P}{\partial r\partial t}\Delta r(\Delta t)+\frac{1}{2}\frac{\partial^2 P}{\partial t^2}(\Delta t)^2+O(\Delta t).
\label{eq:6}
\end{eqnarray}
From Eq. (\ref{eq:3}) and \cite{gu2012time} and \cite{shokrollahi2018subdiffusive}, we have

\begin{eqnarray}
\Delta r&=&\mu_r(\Delta T_{\alpha}(t))+\sigma_rB_2^H(T_{\alpha}(t))\nonumber\\
&=&\mu_r\left(\frac{t^{\alpha-1}}{\Gamma(\alpha)}\right)(\Delta t)+\sigma_r\Delta B_2^H(T_{\alpha}(t))+O(\Delta t).\\
(\Delta r)^2&=&\sigma_r^2\left(\frac{t^{\alpha-1}}{\Gamma(\alpha)}\right)^{2H}(\Delta t)^{2H}+O(\Delta t).\\
\Delta r(\Delta t)&=&O(\Delta t).
\label{eq:7}
\end{eqnarray}
Hence
\begin{eqnarray}
dP(r, t ,T)&=&\left[\mu_r\frac{t^{\alpha-1}}{\Gamma(\alpha)}\frac{\partial P}{\partial r}+\sigma_r^2Ht^{2H-1}\left(\frac{t^{\alpha-1}}{\Gamma(\alpha)}\right)^{2H}\frac{\partial^2 P}{\partial r^2}+\frac{\partial P}{\partial t}\right]dt\nonumber\\
&&+\sigma_r\frac{\partial P}{\partial t}dB_2^H(T_{\alpha}(t)).
\label{eq:8}
\end{eqnarray}
Putting
\begin{eqnarray}
\mu&=&\frac{1}{P}\left[\mu_r\frac{t^{\alpha-1}}{\Gamma(\alpha)}\frac{\partial P}{\partial r}+\sigma_r^2Ht^{2H-1}\left(\frac{t^{\alpha-1}}{\Gamma(\alpha)}\right)^{2H}\frac{\partial^2 P}{\partial r^2}+\frac{\partial P}{\partial t}\right],\nonumber\\
\sigma&=&\frac{1}{P}\left(\frac{\partial P}{\partial r}\right),
\label{eq:9}
\end{eqnarray}
and letting the local expectations hypothesis holds for the term structure of interest rates
(i.e. $\mu=r$), we have
\begin{eqnarray}
&&\frac{\partial P}{\partial t}+\mu_r\frac{t^{\alpha-1}}{\Gamma(\alpha)}\frac{\partial P}{\partial r}\nonumber\\
&&+Ht^{2H-1}\sigma_r^2\left(\frac{t^{\alpha-1}}{\Gamma(\alpha)}\right)^{2H}\frac{\partial^2 P}{\partial r^2}-rP=0.
\label{eq:10}
\end{eqnarray}
Then, zero-coupon bond $P(r, t, T)$ with boundary condition $P(r, t, T)=1$ satisfies the following partial differential equation
\begin{eqnarray}
&&\frac{\partial P}{\partial t}+\mu_r\frac{t^{\alpha-1}}{\Gamma(\alpha)}\frac{\partial P}{\partial r}\nonumber\\
&&+Ht^{2H-1}\sigma_r^2\left(\frac{t^{\alpha-1}}{\Gamma(\alpha)}\right)^{2H}\frac{\partial^2 P}{\partial r^2}-rP=0.
\label{eq:11}
\end{eqnarray}
To solve Equation (\ref{eq:11}) for $P(r, t, T)$, let $\tau =T-t, P(r, t, T) = \exp\{f_1(\tau)-rf_2(\tau)\} $, then  we can get
\begin{eqnarray}
\frac{\partial P}{\partial t}&=&P\left(-\frac{\partial f_1(\tau)}{\partial \tau}+r\frac{\partial f_2(\tau)}{\partial \tau}\right),
\label{eq:12}
\end{eqnarray}
\begin{eqnarray}
\frac{\partial P}{\partial r}&=&-Pf_2(\tau),
\label{eq:13}
\end{eqnarray}
\begin{eqnarray}
\frac{\partial^2 P}{\partial r^2}&=&Pf_2^2(\tau).
\label{eq:14}
\end{eqnarray}
Replacing Equations (\ref{eq:13}) and (\ref{eq:14}) into Equation (\ref{eq:12}) and simplifying Equation (\ref{eq:11}) we get
\begin{eqnarray}
&&P\Bigg[Ht^{2H-1}\sigma_r^2f_2(\tau)^2\left(\frac{t^{\alpha-1}}{\Gamma(\alpha)}\right)^{2H}-\mu_rf_2(\tau)\frac{t^{\alpha-1}}{\Gamma(\alpha)}\nonumber\\&&-\frac{\partial f_1(\tau)}{\partial \tau}
+r\left(\frac{\partial f_2(\tau)}{\partial \tau}-1\right)\Bigg]=0.
\label{eq:15}
\end{eqnarray}
From Equation (\ref{eq:15}), we obtain
\begin{eqnarray}
\frac{\partial f_1(\tau)}{\partial \tau}&=&\sigma_r^2Ht^{2H-1}\left(\frac{t^{\alpha-1}}{\Gamma(\alpha)}\right)^{2H}f_2(\tau)^2-\mu_r\frac{t^{\alpha-1}}{\Gamma(\alpha)}f_2(\tau),\nonumber\\
\frac{\partial f_2(\tau)}{\partial \tau}&=&1.
\label{eq:16}
\end{eqnarray}
Then,
\begin{eqnarray}
f_1(\tau)&=&\frac{H\sigma_r^2}{(\Gamma(\alpha))^{2H}}\int_0^{\tau}(T-v)^{2\alpha H-1}v^2dv\nonumber\\
&&-\frac{\mu_r}{\Gamma(\alpha)}\int_0^{\tau}(T-v)^{\alpha-1}vdv,\label{eq:16-1}\\
f_2(\tau)&=&\tau.
\label{eq:17}
\end{eqnarray}
This ends the proof.
\end{proof}

\begin{figure}[H]
  \centering
          \includegraphics[width=1\textwidth]{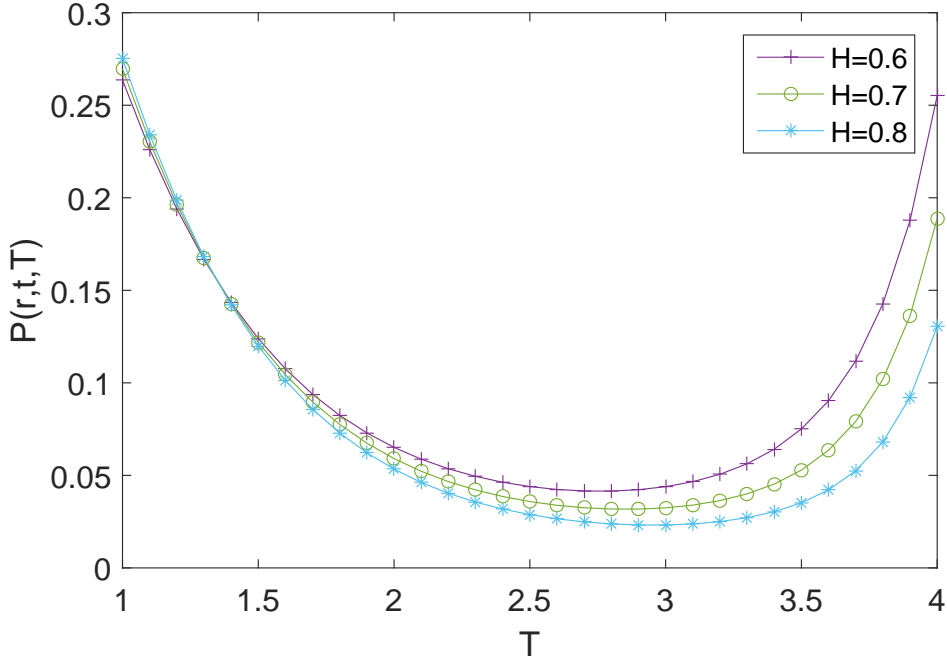}
  \caption{Bond price as a function of maturity time T for different values of $H$, see formula (\ref{eq:18}). The parameters are: $ \mu_r=\sigma_r=r(0)=1, \alpha = 0.9, t=0$.}
\label{fig3}
\end{figure}

\begin{cor}

When $\alpha\uparrow1$, Equations (\ref{eq:1}) and (\ref{eq:3}) reduce to the $FBS$. Then we obtain
\begin{eqnarray}
f_1(\tau)&=&H\sigma_r^2\int_0^{\tau}(T-v)^{2H-1}v^2dv-\mu_r\int_0^{\tau}vdv,
\label{eq:18-1}
\end{eqnarray}
specially, if $t=0$
\begin{eqnarray}
f_1(\tau)&=&\sigma_r^2\frac{T^{2H+2}}{(2H+1)(2H+2)}-\mu_r\frac{T^2}{2},
\label{eq:18-2}
\end{eqnarray}
then

\begin{eqnarray}
P(r, t, T)=\exp\left\{-rT+\sigma_r^2\frac{T^{2H+2}}{(2H+1)(2H+2)}-\mu_r\frac{T^2}{2}\right\}.
\label{eq:18-3}
\end{eqnarray}
\label{cor:1}
\end{cor}

\begin{cor}
If $H=\frac{1}{2}$, from Equation (\ref{eq:16-1}), we obtain

\begin{eqnarray}
f_1(\tau)&=&\frac{1}{2}\frac{\sigma_r^2}{\Gamma(\alpha)}\int_0^{\tau}(T-v)^{\alpha-1}v^2dv\nonumber\\
&&-\frac{\mu_r}{\Gamma(\alpha)}\int_0^{\tau}(T-v)^{\alpha-1}vdv,
\label{eq:19}
\end{eqnarray}
then the result is consistent with the result in \cite{guo2017option}.

Further, if $\alpha\uparrow1$ and $H=\frac{1}{2}$, Equations (\ref{eq:1}) and (\ref{eq:3}) reduce to the geometric Brownian motion, then we have
\begin{eqnarray}
f_1(\tau)=\frac{1}{6}\sigma_r^2\tau^3-\frac{1}{2}\mu_r\tau^2,
\label{eq:20}
\end{eqnarray}
then
\begin{eqnarray}
P(r, t, T)=e^{-r\tau+\frac{1}{6}\sigma_r^2\tau^3-\frac{1}{2}\mu_r\tau^2},
\label{eq:21}
\end{eqnarray}
which is consistent with the result in \cite{kung2009option,cui2010comment}.
\label{cor:2}
\end{cor}
\section{Fractional Black-Scholes equation}\label{sec:2}
This section provides corresponding $FBS$ equation for equity warrants when the stock price and short rate satisfy Eqs. (\ref{eq:1}) and (\ref{eq:3}), respectively. Recall that $B_1^{H}$ and $B_2^H$ are two dependent $FBM$ with Hurst parameter $H \in[\frac{1}{2}, 1)$ and correlation coefficient $\rho$.

\begin{thm} Let assumptions  (i)–(v)  hold and $W_t$ be the valuation of the equity warrant. Then the valuation of equity warrant at time $t\in [0, T]$ satisfies the following $PDE$ and the following boundary condition
\begin{eqnarray}
&&\frac{\partial W}{\partial t}+\widetilde{\sigma}_v^2(t)V^2\frac{\partial^2 W}{\partial V^2}+\widetilde{\sigma}_r^2(t)\frac{\partial^2 W}{\partial r^2}+2\rho\widetilde{\sigma}_r(t)\widetilde{\sigma}_v(t)\frac{\partial^2 W}{\partial V\partial r}\nonumber\\
&&+\mu_r\frac{t^{\alpha-1}}{\Gamma(\alpha)}\frac{\partial W}{\partial r}+rV\frac{\partial W}{\partial V}-rW=0,
\label{eq:22}
\end{eqnarray}
with boundary condition
\begin{eqnarray}
W_T=\frac{1}{N+Mk}(kV_T-NX)^+
\end{eqnarray}
where
\begin{eqnarray}
\widetilde{\sigma}_v^2(t)=Ht^{2H-1}\sigma_v^2\left(\frac{t^{\alpha-1}}{\Gamma(\alpha)}\right)^{2H},
\label{eq:23}
\end{eqnarray}
\begin{eqnarray}
\widetilde{\sigma}_r^2(t)=Ht^{2H-1}\sigma_r^2\left(\frac{t^{\alpha-1}}{\Gamma(\alpha)}\right)^{2H}.
\label{eq:24}
\end{eqnarray}
$\sigma_v, \sigma_r, \mu_r, \mu_v,$ are constant, $H\in[\frac{1}{2}, 1)$ and $\alpha\in (\frac{1}{2}, 1)$.
\label{th:1}
\end{thm}
\begin{proof}
It is worth pointing out that $W_t$ is a function of the current value of the underlying asset $V_t$, the stochastic interest rate $r_t$ and time $t$: $W_t=W(V_t, r_t, t)$. Then, at the expiry date $t=T$, we have
\begin{eqnarray}
W_T=\frac{1}{N+Mk}(kV_T-NX)^+,
\end{eqnarray}
which is the boundary condition of Eq \ref{eq:22}.
Consider a portfolio with one share of warrants, $D_{1t}$ shares of stock and $D_{2t}$ shares of zero-coupon bond $P(r, t, T)$. Then
the value of the portfolio $\Pi$ at current time $t$ is
\begin{eqnarray}
\Pi_t=W_t-D_{1t}V_t-D_{2t}P_t.
\label{eq:25}
\end{eqnarray}
Then, from \cite{guo2017option} we have
\begin{eqnarray}
d\Pi_t&=&dW_t-D_{1t}dV_t-D_{2t}dP_t\nonumber\\
&=&\Bigg[\frac{\partial W}{\partial t}dt+Ht^{2H-1}\sigma_v^2V_t^2\left(\frac{t^{\alpha-1}}{\Gamma(\alpha)}\right)^{2H}\frac{\partial^2 W}{\partial V^2}+Ht^{2H-1}\sigma_r^2\left(\frac{t^{\alpha-1}}{\Gamma(\alpha)}\right)^{2H}\frac{\partial^2 W}{\partial r^2}\nonumber\\
&+&2Ht^{2H-1}\rho\sigma_r\sigma_vV\left(\frac{t^{\alpha-1}}{\Gamma(\alpha)}\right)^{2H}\frac{\partial^2 W}{\partial V\partial r}\Bigg]dt+\Bigg[\frac{\partial W}{\partial t}-D_{1t}\Bigg]dV_t\nonumber\\
&+&\Bigg[\frac{\partial W}{\partial r}-D_{2t}\frac{\partial P}{\partial r}\Bigg]dr+D_{2t}\Bigg[\frac{\partial P}{\partial t}+Ht^{2H-1}\sigma_r^2\left(\frac{t^{\alpha-1}}{\Gamma(\alpha)}\right)^{2H}\frac{\partial^2 P}{\partial r^2}\Bigg]dt.
\label{eq:26}
\end{eqnarray}
Set $D_{1t}=\frac{\partial W}{\partial V},\,D_{2t}=\frac{\frac{\partial W}{\partial r}}{\frac{\partial P}{\partial r}},$ to eliminate the stochastic noise. Then
\begin{eqnarray}
d\Pi_t&=&\nonumber\\
&=&\Bigg[\frac{\partial W}{\partial t}+Ht^{2H-1}\left(\frac{t^{\alpha-1}}{\Gamma(\alpha)}\right)^{2H}\left(\sigma_v^2V^2\frac{\partial^2 V}{\partial V^2}+\sigma_r^2\frac{\partial^2 W}{\partial r^2}+2\rho\sigma_r\sigma_vV\frac{\partial^2 W}{\partial V\partial r}\right)\Bigg]dt\nonumber\\
&-&\frac{\frac{\partial W}{\partial r}}{\frac{\partial P}{\partial r}}\Bigg[rP-\mu_r\frac{t^{\alpha-1}}{\Gamma(\alpha)}\frac{\partial P}{\partial r}\Bigg]dt.
\label{eq:27}
\end{eqnarray}

The return of the amount $\Pi_t$ invested in bank account
is equal to $r(t)\Pi_tdt$ at time $dt$, $\E(d\Pi_t)=r(t)\Pi_tdt=r(t)\left(W_t-D_{1t}V_t-D_{2t}P_t\right)$, hence from Equation (\ref{eq:27}) we have

\begin{eqnarray}
&&\frac{\partial W}{\partial t}+Ht^{2H-1}\left(\frac{t^{\alpha-1}}{\Gamma(\alpha)}\right)^{2H}\left(\sigma_v^2V^2\frac{\partial^2 C}{\partial S^2}+\sigma_r^2\frac{\partial^2 W}{\partial r^2}+2\rho\sigma_r\sigma_vV\frac{\partial^2 W}{\partial V\partial r}\right)\nonumber\\
&&+\mu_r\frac{t^{\alpha-1}}{\Gamma(\alpha)}\frac{\partial W}{\partial r}+rv\frac{\partial W}{\partial V}-rW=0.
\label{eq:28}
\end{eqnarray}
Let
\begin{eqnarray}
\widetilde{\sigma}_v^2(t)=Ht^{2H-1}\sigma_v^2\left(\frac{t^{\alpha-1}}{\Gamma(\alpha)}\right)^{2H},
\label{eq:29}
\end{eqnarray}
\begin{eqnarray}
\widetilde{\sigma}_r^2(t)=Ht^{2H-1}\sigma_r^2\left(\frac{t^{\alpha-1}}{\Gamma(\alpha)}\right)^{2H}.
\label{eq:30}
\end{eqnarray}
Then
\begin{eqnarray}
&&\frac{\partial W}{\partial t}+\widetilde{\sigma}_v^2(t)V_t^2\frac{\partial^2 W}{\partial V_t^2}+\widetilde{\sigma}_r^2(t)\frac{\partial^2 W}{\partial r^2}+2\rho\widetilde{\sigma}_r(t)\widetilde{\sigma}_v(t)\frac{\partial^2 W}{\partial V\partial r}\nonumber\\
&&+\mu_r\frac{t^{\alpha-1}}{\Gamma(\alpha)}\frac{\partial W}{\partial r}+rV\frac{\partial W}{\partial V}-rW=0,
\label{eq:31}
\end{eqnarray}
and the proof is completed.
\end{proof}

\begin{thm}
Let assumptions  (i)–(v)  hold and $W_t$ be the value of the equity warrant. Then the value of an equity warrant at time $t\in [0, T]$ is given by
\begin{eqnarray}
W_t&=&\frac{1}{N+Mk}\left(kV_t\phi(d_1)-NXe^{-r(T-t)}P(r, t, T)\phi(d_2)\right),\\
\label{eq:35}
\end{eqnarray}
where
\begin{eqnarray}
d_1&=&\frac{\ln\frac{kV_t}{NX}-\ln P(r, t, T)+\frac{H}{(\Gamma(\alpha))^{2H}}\int_t^T\widehat{\sigma}^2(v)v^{2\alpha H-1}dv}{\sqrt{\frac{2H}{(\Gamma(\alpha))^{2H}}\int_t^T\widehat{\sigma}^2(v)v^{2\alpha H-1}dv}},
\label{eq:36}\\
d_2&=&d_1-\sqrt{\frac{2H}{(\Gamma(\alpha))^{2H}}\int_t^T\widehat{\sigma}^2(v)v^{2\alpha H-1}dv},\\
\widehat{\sigma}^2(t)&=&\sigma_v^2+2\rho\sigma_r\sigma_v(T-t)+\sigma_r^2(T-t)^2.
\label{eq:37}
\end{eqnarray}
$P(r, t, T)$ is given by Eq (\ref{eq:18}) and $\phi(.)$ is the cumulative standard normal distribution function.
\label{th:2}
\end{thm}

\begin{proof}

Consider the partial differential equation (\ref{eq:22}) of the equity warrants with boundary condition $W_T=\frac{1}{N+Mk}(kV_T-NX)^+$

\begin{eqnarray}
&&\frac{\partial W}{\partial t}+\widetilde{\sigma}_v^2(t)V^2\frac{\partial^2 W}{\partial V^2}+\widetilde{\sigma}_r^2(t)\frac{\partial^2 W}{\partial r^2}+2\rho\widetilde{\sigma}_r(t)\widetilde{\sigma}_v(t)\frac{\partial^2 W}{\partial V\partial r}\nonumber\\
&&+\mu_r\frac{t^{\alpha-1}}{\Gamma(\alpha)}\frac{\partial W}{\partial r}+rV\frac{\partial W}{\partial V}-rW=0,
\label{eq:38}
\end{eqnarray}
Denote
\begin{eqnarray}
z=\frac{V}{P(r, t, T)},\quad\Theta(z, t)=\frac{W(V, r, t)}{P(r, t, T)}.
\label{eq:39}
\end{eqnarray}
Then we get

\begin{eqnarray}
\frac{\partial W}{\partial t}&=&\Theta\frac{\partial P}{\partial t}+P\frac{\partial \Theta}{\partial t}-z\frac{\partial \Theta}{\partial z}\frac{\partial P}{\partial t},\nonumber\\
\frac{\partial W}{\partial r}&=&\Theta\frac{\partial P}{\partial r}-z\frac{\partial \Theta}{\partial z}\frac{\partial P}{\partial r},\nonumber\\
\frac{\partial W}{\partial V}&=&\frac{\partial \Theta}{\partial z},\label{eq:40}\\
\frac{\partial^2 W}{\partial r^2}&=&\Theta\frac{\partial^2 P}{\partial r^2}-z\frac{\partial \Theta}{\partial z}\frac{\partial^2 P}{\partial r^2}+\frac{z^2}{P}\frac{\partial^2\Theta}{\partial z^2}\left(\frac{\partial P}{\partial r}\right)^2,\nonumber\\
\frac{\partial^2 W}{\partial r\partial V}&=&-\frac{z}{P}\frac{\partial^2 \Theta}{\partial z^2}\frac{\partial P}{\partial r},\nonumber\\
\frac{\partial^2 W}{\partial V^2}&=&\frac{1}{P}\frac{\partial^2 \Theta}{\partial z^2}.\nonumber
\end{eqnarray}
Inserting Equation (\ref{eq:40}) into Equation (\ref{eq:38})
\begin{eqnarray}
\frac{\partial \Theta}{\partial t}&+&\frac{\partial^2 \Theta}{\partial z^2}\left[\widetilde{\sigma}_v^2(t)\frac{V^2}{P^2}+2\rho z^2 \widetilde{\sigma}_r(t)\widetilde{\sigma}_v(t)\frac{1}{P}\frac{\partial P}{\partial r}+\widetilde{\sigma}_r^2(t)z^2\left(\frac{1}{P}\frac{\partial P}{\partial r}\right)^2\right]\nonumber\\
&-&\frac{z}{P}\left[\frac{\partial P}{\partial t}+\widetilde{\sigma}_r^2(t)\frac{\partial^2 P}{\partial r^2}+\mu_r\frac{t^{\alpha-1}}{\Gamma(\alpha)}\frac{\partial P}{\partial r}-r\frac{V}{z}\right]\nonumber\\
&+&\frac{\Theta}{P}\left[\frac{\partial P}{\partial t}+\widetilde{\sigma}_r^2(t)\frac{\partial^2 P}{\partial r^2}+\mu_r\frac{t^{\alpha-1}}{\Gamma(\alpha)}\frac{\partial P}{\partial r}-rP\right]=0.
\label{eq:41}
\end{eqnarray}
From Equation (\ref{eq:11}), we can obtain
\begin{eqnarray}
\frac{\partial \Theta}{\partial t}+\overline{\sigma}^2(t)z^2\frac{\partial^2 \Theta}{\partial z^2}=0,
\label{eq:42}
\end{eqnarray}
with boundary condition $\Theta(z,T)=(z-K)^+$,

where
\begin{eqnarray}
\overline{\sigma}^2(t)=\widetilde{\sigma}_v^2(t)+2\rho\widetilde{\sigma}_r(t)\widetilde{\sigma}_v(t)(T-t)+\widetilde{\sigma}_r(t)^2(T-t)^2.
\label{eq:43}
\end{eqnarray}

The solution of partial differential Equation (\ref{eq:42}) with boundary condition $\Theta(z,T)=\frac{1}{N+Mk}(kz-NX)^+$, is given by

\begin{eqnarray}
\Theta(z,t)=kz\phi(\widehat{d}_1)-NX\phi(\widehat{d}_2),
\label{eq:44}
\end{eqnarray}
here
\begin{eqnarray}
\widehat{d}_1&=&\frac{\ln\frac{kz}{NX}+\int_t^T\overline{\sigma}^2(v)ds}{\sqrt{2\int_t^T\widehat{\sigma}^2(v)dv}},\\
\widehat{d}_2&=&\widehat{d}_1-{\sqrt{2\int_t^T\overline{\sigma}^2(v)dv}}.
\label{eq:45}
\end{eqnarray}
Thus, from Equations (\ref{eq:39}) and (\ref{eq:44})--(\ref{eq:45}) we obtain
\begin{eqnarray}
W(V, r, t)&=&\frac{1}{N+Mk}\left(kV_t\phi(d_1)-NXe^{-r(T-t)}P(r, t, T)\phi(d_2)\right)
\label{eq:46}
\end{eqnarray}
where
\begin{eqnarray}
d_1&=&\frac{\ln\frac{kV_t}{NX}-\ln P(r, t, T)+\frac{H}{\left(\Gamma(\alpha)\right)^{2H}}\int_t^T\widehat{\sigma}^2(v)v^{2\alpha H-1}dv}{\sqrt{\frac{2H}{\left(\Gamma(\alpha)\right)^{2H}}\int_t^T\widehat{\sigma}^2(v)v^{2\alpha H-1}dv}},\\
d_2&=&d_1-\sqrt{\frac{2H}{\left(\Gamma(\alpha)\right)^{2H}}\int_t^T\widehat{\sigma}^2(v)v^{2\alpha H-1}dv}.
\label{eq:47}
\end{eqnarray}
\end{proof}

\begin{figure}[H]
  \centering
          \includegraphics[width=1\textwidth]{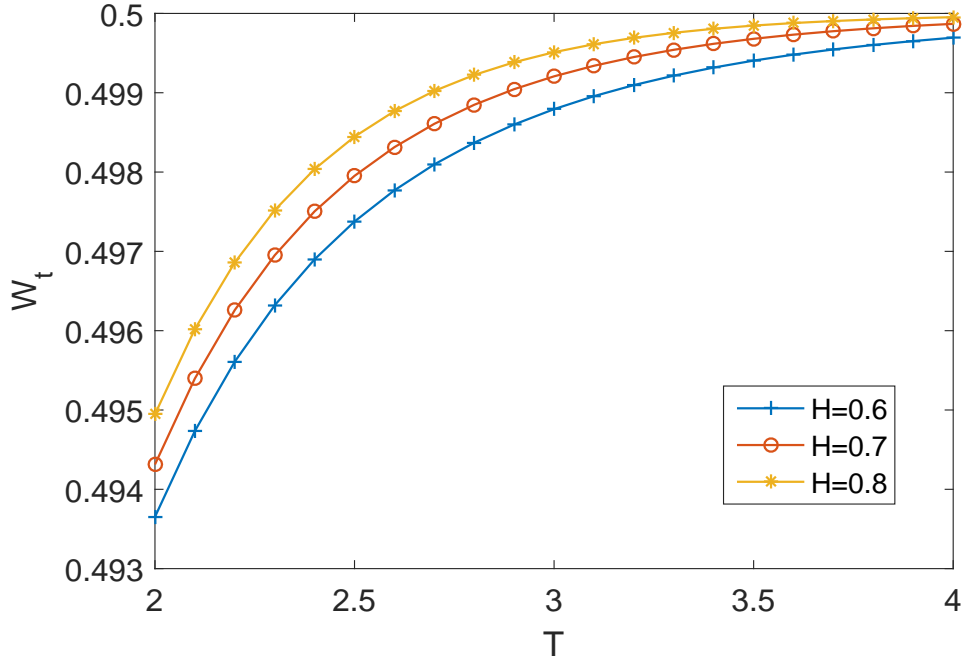}
  \caption{Value of equity warrant as a function of maturity time $T$ for different values of $H$, see formula (\ref{eq:35}). The parameters are: $\mu_v=\sigma_v=\mu_r=\sigma_r=r(0)=V(0)=N=M=k=X=1, \alpha = 0.9, \rho=0.5, t=0.$
}
\label{fig4}
\end{figure}

\section*{Acknowledgements}
The research of MM was partially supported by NCN SONATA BIS-9 grant number 2019/34/E/ST1/00360.

\bibliographystyle{siam}
\bibliography{reference}

\end{document}